\documentclass[runningheads]{llncs}
\usepackage{graphicx}
\usepackage{amsmath,amsfonts,amssymb}
\usepackage{graphicx}
\usepackage{algorithm}
\usepackage{algpseudocode}

\newcommand{\ket}[1]{|#1\rangle}
\begin{document}
\title{Quantum Algorithm for the Shortest Superstring Problem\thanks{A part of the reported study is funded by RFBR according to the research project No.20-37-70080. The research is funded by the subsidy allocated to Kazan Federal University for the state assignment in the sphere of scientific activities, project No. 0671-2020-0065.}}
%
%
\author{Kamil  Khadiev\inst{1,2} \and
Carlos Manuel Bosch Machado\inst{1}}
\authorrunning{K. Khadiev, C.M. Bosh-Machado}
%
\institute{Institute of Computational Mathematics and Information Technologies, Kazan Federal University, Kremlyovskaya, 35, Kazan, Russia\and
Zavoisky Physical-Technical Institute, FRC Kazan Scientific Center of RAS, Kazan, Russia}
\maketitle              
\begin{abstract}
In this paper, we consider the ``Shortest Superstring Problem''(SSP) or the ``Shortest Common Superstring Problem''(SCS). The problem is as follows. For a positive integer $n$, a sequence of n strings $S=(s^1,\dots,s^n)$ is given. We should construct the shortest string $t$ (we call it superstring) that contains each string from the given sequence as a substring.
The problem is connected with the sequence assembly method for reconstructing a long DNA sequence from small fragments. We present a quantum algorithm with running time $O^*(1.728^n)$. Here $O^*$ notation does not consider polynomials of $n$ and the length of $t$. 

\end{abstract}
\section{Introduction}
\label{sec:intro}  

In this paper, we consider the ``Shortest Superstring Problem''(SSP) or the ``Shortest Common Superstring Problem''(SCS). The problem is as follows. For a positive integer $n$, a sequence of n strings $S=(s^1,\dots,s^n)$ is given. We should construct the shortest string t (we call it superstring) that contains each string from the given sequence as a substring.
The problem is connected with the sequence assembly method for reconstructing a long DNA sequence from small fragments \cite{msdd2000}. There are two types of sequence assemble problems. The first one is the Reference-guided genome assembly method that constructs an existing long DNA string from the sequence $S$. For the problem, we should know the string $t$ apriori and check whether we can construct it from $S$. The second one is de-novo assembly; in this problem, we do not have the string $t$ and should construct it using all strings from $S$. The Shortest Superstring Problem is used as one of the main tools for solving de-novo assembly problems \cite{bdfhw95}. The problem has applications in other areas such as virology and immunology (the SSP models the compression of viral genome); the SSP can be used to achieve data compression; in scheduling ( solutions can be used to schedule operations in machines with coordinated starting times), and others. It is known that the Shortest Superstring Problem is an NP-hard problem \cite{m78,m94,m98,v2005}. Researchers interested in approximation algorithm, the best known algorithm is \cite{ks2013}. At the same time, researchers are interested in exact solutions also. The algorithm based on \cite{b62,hk62} have $O^*(2^n)$ running time. If we have a restriction on a length of the strings $s^i$, then there are better algorithms. If a length of strings $s^i$ is at most $3$, then the there is an algorithm \cite{gkm2013} with running time $O^*(1.443^n)$. For a constant $c$, if a length of strings $s^i$ is at most $c$, then there is a randomized algorithm \cite{gkm2014} with running time $O^*(2^{(1-f(c))c})$ where $f(c)= 1/(1 + 2c^2)$.

We investigate the problem from the quantum computing \cite{nc2010,a2017} point of view. There are many problems where quantum algorithms outperform the best-known classical algorithms. Some of them can be founded here \cite{dw2001,quantumzoo,ks2019,kksk2020,kks2019,gnbk2021}. 
Problems for strings are examples of such problems \cite{ki2019,kk2021,rv2003,l2020,aj2021,m2017}. One of the most important performance metrics in this regard is \emph{query complexity}; and we investigate problems using this metric for complexity.
Researchers investigated quantum algorithms for different problems that can be represented as an assembly of a long string $t$ from the collection of strings $S$ if $t$ is known a priori \cite{kr2021a,kr2021b}.

We present a quantum algorithm for SSP with running time $O^*(1.728^n)$. Here $O^*$ notation does not consider a log factor; in other words, polynomials of $n$ and the length of $t$. The algorithm is based on Grover’s search algorithm \cite{g96} and the Dynamic programming approach for a Boolean cube \cite{abikpv2019,b62,hk62}. As far as we know, our algorithm is the first quantum algorithm for SSP.

The structure of this paper is the following. Section \ref{sec:prelims} contains preliminaries. We present the algorithm in Section \ref{sec:algo}. Section \ref{sec:conclusion} concludes the paper.

\section{Preliminaries}\label{sec:prelims}

Let us consider a string $u=(u_1,\dots,u_m)$. Let $u[i,j]$ be a substring $(u_i,\dots,u_j)$ for $1\leq i\leq j\leq m$. Let $|u|$ be a length of a string $u$.

\subsection{Shortest Superstring Problem} The problem is as follows. For positive integers $n$ and $k$, a sequence of $n$ strings $S=(s^1,\dots,s^n)$ is given. The length of each string $|s^i|=k$ for $i\in\{1,\dots,n\}$.  We should construct the shortest string $t$ (we call it superstring), i.e. $|t|$ is the minimal possible such that each $s^i$ is substring of $t$ for $i\in\{1,\dots,n\}$. In there words, for each $i\in\{1,\dots,n\}$ there are $1\leq \ell,r\leq |t|$ such that $t[\ell,r]=s^i$.  

\subsection{Quantum Query Model}
We use the standard form of the quantum query model. 
Let $f:D\rightarrow \{0,1\},D\subseteq \{0,1\}^N$ be an $N$ variable function. An input for the function is $x=(x_1,\dots,x_N)\in D$ where $x_i\in\{0,1\}$ for $i\in\{1,\dots,N\}$.

We are given oracle access to the input $x$, i.e. it is realized by a specific unitary transformation usually defined as $\ket{i}\ket{z}\ket{w}\rightarrow \ket{i}\ket{z+x_i\pmod{2}}\ket{w}$ where the $\ket{i}$ register indicates the index of the variable we are querying, $\ket{z}$ is the output register, and $\ket{w}$ is some auxiliary work-space. It can be interpreted as a sequence of control-not transformations such that we apply inversion operation (X-gate) to the second register that contains $\ket{z}$ in a case of the first register equals $i$ and the variable $x_i=1$. We interpret the oracle access transformation as $N$ such controlled transformations for each $i\in\{1,\dots,N\}$.   

An algorithm in the query model consists of alternating applications of arbitrary unitaries independent of the input and the query unitary, and a measurement in the end. The smallest number of queries for an algorithm that outputs $f(x)$ with a probability that is at least $\frac{2}{3}$ on all $x$ is called the quantum query complexity of the function $f$ and is denoted by $Q(f)$. We refer the readers to \cite{nc2010,a2017} for more details on quantum computing. 

In this paper's quantum algorithms, we refer to the quantum query complexity as the quantum running time. We use modifications of Grover's search algorithm \cite{g96,bbht98} as quantum subroutines. For these subroutines, time complexity is more than query complexity for additional log factor.

\section{Algorithm}\label{sec:algo}

We discuss our algorithm in this section.

Firstly, let us reformulate the problem in the Graph form.

Let us construct a full directed weighted graph $G=(V,E)$ by the sequence $S$. A vertex $v^i$ corresponds to the string $s^i$ for $i\in\{1,\dots,n\}$. Set of vertexes is $V=(v^1,\dots,v^n)$. The weight of an edge between two vertexes $v^i$ and $v^j$ is the length of maximal overlap for $s^i$ and $s^j$. Formally, $w(i,j)=\max\{z: s^i[k-z+1,k]=s^j[1,z]\}$. We can see that any path that visits all vertexes exactly once represents a superstring. Let the weight of the path be the sum of weights of all edges the belongs to a path. The path that visits all vertexes exactly once and has maximal weight represents the shortest superstring.
Let $P=(v^{i_1},\dots,v^{i_\ell})$ be a path. Let a weight of the path $P$ be $w(P)=w(v^{i_1},v^{i_2})+\dots+w(v^{i_{\ell-1}},v^{i_{\ell}})$, let $|P|=\ell$.

Let us present two procedures:
\begin{itemize}
    \item $\textsc{ConstructTheGraph}(S)$ constructs the graph $G=(V,E)$ by $S$. The implementation of the procedure is presented in Algorithm \ref{alg:constr}. 
    \item $\textsc{ConstructSuperstringByPath}(P)$ constructs the target superstring by a path $P$ in the graph $G=(V,E)$. Implementation of the procedure is presented in Algorithm \ref{alg:getpath}.
\end{itemize}
\begin{algorithm}
\caption{Implementation of $\textsc{ConstructTheGraph}(S)$ for $S=(s^1,\dots,s^n)$, $|s^i|=k$.}\label{alg:constr}
\begin{algorithmic}
\State $V=(v^1,\dots,v^n)$
\For{$i\in\{1,\dots,n\}$}
\For{$j\in\{1,\dots,n\}$}
\If{$i\neq j$}
\State $maxOverlap\gets 0$
\For{$z\in\{1,\dots,k\}$}
\If{$s^i[n-z+1,n]=s^j[1,z]$}
\State $maxOverlap\gets z$
\EndIf
\EndFor
\State $E\gets E\cup (v^i,v^j)$
\State $w(v^i,v^j)\gets maxOverlap$
\EndIf
\EndFor
\EndFor
\State \Return{$(V,E)$}
\end{algorithmic}
\end{algorithm}

\begin{algorithm}
\caption{Implementation of $\textsc{ConstructSuperstringByPath}(P)$ for $P=(v^{i_1},\dots,v^{i_\ell})$.}\label{alg:getpath}
\begin{algorithmic}
\State $t=s^{i_1}$
\For{$j\in\{2,\dots,\ell\}$}
\State $t\gets t\circ s^{i_j}[w(v^{i_{j-1}},v^{i_{j-1}})+1,k]$\Comment{Here $\circ$ is the concatenation operation.}
\EndFor
\State \Return{$t$}
\end{algorithmic}
\end{algorithm}

Let us consider a function $L:2^V\times V\times V\to \mathbb{R}$ where $2^V$ is the set of all subsets of $V$. The function $L$ is such that $L(S,v,u)$ is the maximum of all weights of paths that visit all vertexes of $S$ exactly ones, start from the vertex $v$, and finish in the vertex $u$. If there is no such path, then we assume that $L(S,v,u)=-\infty$.

Let the function $F:2^V\times V\times V\to V^*$ be such that $F(S,v,u)$ is the path that visit all vertexes of $S$ exactly once, starts from the vertex $v$, finishes in the vertex $u$ and has maximal weight. In other words for $P=F(S,u,v)$ we have $w(P)=L(S,u,v)$.

We assume, that $L(\{v\},v,v)=0$  and $F(\{v\},v,v)=(v)$ for any $v\in V$ by definition.

Let us discuss the properties of the function.

\begin{property} Suppose $S\subset V, v,u\in V$, an integer $k\leq |S|$. The function $L$ is such that
\[L(S,v,u)=\max\limits_{S'\subset S,|S'|=k,y\in S'}\left(L(S',v,y)+L((S\backslash S') \cup \{y\},y,u)\right)\]
and $F(S,u,v)$ is the path that is concatenation of corresponding paths.
\end{property}
\begin{proof}
Let $P^1=F(S',v,y)$ and $P^2=F((S\backslash S') \cup \{y\},y,u)$. The path $P=P^1\circ P^2$ belongs to $S'$, starts from $v$ and finishes in $u$, where $\circ$ means concatenation of paths excluding the duplication of common vertex $v$. Because of definition of $L$, we have $L(S,v,u)\geq w(P)$.

Assume that there is a path $T=(v^{i_1},\dots,v^{i_\ell})$ such that $w(T)=L(S,v,u)$ and $w(T)>w(P)$. Let us select $S''$ such that $|S''|=k, S''\subset S$ and there is $j<|T|$ such that $R^1=v^{i_1},\dots,v^{i_j}\in S''$ and $R^2=v^{i_j},v^{i_{j+1}},\dots,v^{i_\ell}\not\in S''\backslash\{v_j\}$. Then $w(R^1)\leq w(P^1)$ and $w(R^2)\leq w(P^2)$ by definition of $L$. Therefore, $w(R)=w(R^1)+w(R^2)-1\leq w(P^1)+w(P^2)-1=w(P)$. We obtain a contradiction with assumption.
\end{proof}

As a corollary we obtain the following result:
\begin{corollary}\label{cor:one-edge}
Suppose $S\subset V, v,u\in V$, ${\cal I}(u)$ is the set of all neighbor vertexes of $u$. The function $L$ is such that
\[L(S,v,u)=\max\limits_{y\in S\backslash\{u\}, y\in {\cal I}(u)}\left(L( S\backslash\{u\},v,y)+w(y,u)\right).\]
and $F(S,u,v)$ is the path that is the corresponding path.
\end{corollary}

Using this idea, we construct the following algorithm.

{\bf Step 1}. Let $\alpha=0.055$. We classically compute $L(S,v,u)$ and $F(S,v,u)$ for all $|S|\leq (1-\alpha)\frac{n}{4}$ and $v,u\in S$ 

{\bf Step 2}. Let $V_4\subset V$ be such that $|V_4|=\frac{n}{4}$. Then, we have 

\[L(V_4,u,v)=\max\limits_{V_{\alpha}\subset V_4,|V_{\alpha}|=(1-\alpha)n/4,y\in V_{\alpha}}\left(L(V_{\alpha},v,y)+L((V_4\backslash V_{\alpha}) \cup \{y\},y,u)\right).\]

Let $V_2\subset V$ be such that $|V_2|=\frac{n}{2}$. Then, we have

\[L(V_2,u,v)=\max\limits_{V_4\subset V_2,|V_4|=n/4,y\in V_4}\left(L(V_4,v,y)+L((V_2\backslash V_4) \cup \{y\},y,u)\right).\]

Finally,
\[L(V,u,v)=\max\limits_{V_2\subset V,|V_2|=n/2,y\in V_2}\left(L(V_2,v,y)+L((V\backslash V_2) \cup \{y\},y,u)\right).\]

We can compute $L(V,u,v)$ and corresponding $F(V,u,v)$ using three nested procedures for maximum finding. As such procedure, we use Durr-Hoyer \cite{dh96,dhhm2004} quantum minimum finding algorithm.

Note that the error probability for the Durr-Hoyer algorithm is at most $0.1$. So, we use the standard boosting technique to decrease the total error probability to constant by $O(n)$ repetition of the maximum finding algorithm in each level.

Let us present the implementation of Step 1. Assume that ${\cal I}(u)$ is the sequence of neighbor vertexes for $u$.
Let us present a recursive function $\textsc{GetL}(S,v,u)$ for $S\subset V,u,v\in V$ with cashing that is Dynamic Programming approach in fact. The function is based on Corollary \ref{cor:one-edge}.

\begin{algorithm}
\caption{$\textsc{GetL}(S,v,u)$.}
\begin{algorithmic}
\If{$v=u$ and $S=\{v\}$}\Comment{Initialization}
\State $L(\{v\},v,v)\gets 0$
\State $F(\{v\},v,v)\gets (v)$
\EndIf
\If {$L(S,v,u)$ is not computed}
\State $weight\gets -\infty$
\State $path\gets ()$
\For{$y \in {\cal I}(u)$}
\If{$y\in S\backslash\{u\}$ and $\textsc{GetL}( S\backslash\{u\},v,y)+w(y,u)>weight$}
\State $weight\gets L( S\backslash\{u\},v,y)+w(y,u)$
\State $path\gets F( S\backslash\{u\},v,y)\cup u$
\EndIf
\EndFor
\State $L(S,v,u)\gets weight$
\State $F(S,v,u)\gets path$
\EndIf
\State \Return{$L(S,v,u)$}
\end{algorithmic}
\end{algorithm}

\begin{algorithm}
\caption{$\textsc{Step1}$.}
\begin{algorithmic}
\For{$S \in 2^V$ such that $|S|\leq (1-\alpha)\frac{n}{4}$}
\For{$v\in V$}
\For{$u\in V$}
\If{$v \in S$ and $u\in S$}
\State $\textsc{GetL}(S,v,u)$\Comment{We are computing $L(S,v,u)$ and $F(S,v,u)$ but we are not needing these results at the moment. We need it for Step 2.}
\EndIf
\EndFor
\EndFor
\EndFor
\end{algorithmic}
\end{algorithm}

Let $\textsc{QMax}((x_1,\dots,x_N))$ be the implementation of the quantum maximum finding algorithm \cite{dh96} for a sequence $x_1,\dots,x_N$.

The most nested quantum maximum finding algorithm for some $V_4\subset V, |V_4|=\frac{n}{4}$ and $u,v\in V$ is \[\textsc{QMax}((L(V_{\alpha},v,y)+L((V_4\backslash V_{\alpha}) \cup \{y\},y,u):V_{\alpha}\subset V_4,|V_{\alpha}|=(1-\alpha)\frac{n}{4},y\in V_{\alpha}))\]

The middle quantum maximum finding algorithm for some $V_2\subset V, |V_2|=\frac{n}{2}$ and $u,v\in V$ is 

\[\textsc{QMax}((L(V_{4},v,y)+L((V_2\backslash V_{4}) \cup \{y\},y,u):V_4\subset V_2,|V_4|=n/4,y\in V_4))\]

Note that $|V_4|=n/4$ and $|V_2\backslash V_{4}|=n/4$. We use the invocation of $\textsc{QMax}$ (the most nested quantum maximum finding algorithm) instead of $L(V_{4},v,y)$  and $L(V_2\backslash V_{4},y,u)$.

The final quantum maximum finding algorithm for some  $u,v\in V$ is 

\[\textsc{QMax}((L(V_2,v,y)+L((V\backslash V_2) \cup \{y\},y,u):V_2\subset V,|V_2|=n/2,y\in V_2))\]

Note that $|V_2|=n/2$ and $|V\backslash V_2|=n/2$. We use the invocation of $\textsc{QMax}$ (the middle quantum maximum finding algorithm) instead of $L(V_2,v,y)$ and $L((V\backslash V_2) \cup \{y\},y,u)$.

The procedure $\textsc{QMax}$ returns not only the maximal value, but the index of the target element. Therefore, by the ``index'' we can obtain the target paths using $F$ function. So, the result path is
$P=P^1\circ P^2$, where $P^1$ is the result path for $L(V_2,v,y)$ and $P^2$ is the result path for $L((V\backslash V_2) \cup \{y\},y,u)$.

$P^1=P^{1,1}\circ P^{1,2}$, where $P^{1,1}$ is the result path for $L(V_{4},v,y)$ and $P^{1,2}$ is the result path for $L((V_2\backslash V_{4}) \cup \{y\},y,u)$. By the same way we can construct $P^2=P^{2,1}\circ P^{2,2}$.

$P^{1,1}=P^{1,1,1}\circ P^{1,1,2}$, where $P^{1,1,1}$ is the result path for $L(V_{\alpha},v,y)$ and $P^{1,1,2}$ is the result path for $L((V_4\backslash V_{\alpha}) \cup \{y\},y,u)$. Note, that this values were precomputed classically on Step 1, and were have stored them in $F(V_{\alpha},v,y)$ and $F((V_4\backslash V_{\alpha}) \cup \{y\},y,u)$ respectively.

By the same way we can construct
\[P^{1,2}=P^{1,2,1}\circ P^{1,2,2},\]
\[P^{2,1}=P^{2,1,1}\circ P^{2,1,2},\]
\[P^{2,2}=P^{2,2,1}\circ P^{2,2,2}.\]

The final Path is
\[P=P^1\circ P^2=(P^{1,1}\circ P^{1,2})\circ(P^{2,1}\circ P^{2,2})=\]
\[\Big((P^{1,1,1}\circ P^{1,1,2})\circ (P^{1,2,1}\circ P^{1,2,2})\Big)\circ\Big((P^{2,1,1}\circ P^{2,1,2})\circ (P^{2,2,1}\circ P^{2,2,2})\Big)\]

Let us present the final algorithm as Algorithm \ref{alg:main}.
\begin{algorithm}
\caption{Algorithm for SSP.}\label{alg:main}
\begin{algorithmic}
\State $(V,E)\gets\textsc{ConstructTheGraph}(S)$
\State $\textsc{Step1}()$
\State $weight\gets -\infty$
\State $path \gets ()$
\For{$v\in V$}
\For{$u\in V$}
\State $currentWeight\gets\textsc{QMax}((L(V_2,v,y)+L((V\backslash V_2) \cup \{y\},y,u):V_2\subset V,|V_2|=n/2,y\in V_2))$
\If{$weight<currentWeight$}
\State $weight\gets currentWeight$
\State $path\gets \Big((P^{1,1,1}\circ P^{1,1,2})\circ (P^{1,2,1}\circ P^{1,2,2})\Big)\circ\Big((P^{2,1,1}\circ P^{2,1,2})\circ (P^{2,2,1}\circ P^{2,2,2})\Big)$
\EndIf
\EndFor
\EndFor
\State $t\gets \textsc{ConstructSuperstringByPath}(path)$
\State \Return{t}
\end{algorithmic}
\end{algorithm}

The complexity of the algorithm is presented in the following theorem.
\begin{theorem}
 Algorithm \ref{alg:main} solves LTP with $O^*\left(1.728^n\right)$ running time and constant bounded error.
\end{theorem}
\begin{proof}
The correctness of the algorithm follows from the above discussion. Let us present an analysis of running time.

Complexity of Step 1 (Classical preprocessing) is \[\sum_{i=1}^{(1-\alpha) \frac{n}{4}}O^*\left( \binom{i}{(1-\alpha) \frac{n}{4}}\right)=O^*(1.728^n).\]

Complexity of Step 2 (Quantum part) is complexity of three nested Durr-Hoyer maximum finding algorithms. Due to  \cite{dh96,g96,dhhm2004}, the complexity is

\[O^*\left(\sqrt{\binom{m}{n/2}}\cdot\sqrt{\binom{n/2}{n/4}}\cdot\sqrt{\binom{n/4}{\alpha n/4}}\right)=O^*(1.728^n).\]

Complexity of $\textsc{ConstructTheGraph}$ is $O(n^2k)$ and complexity of $\textsc{ConstructSuperstringByPath}$ is $O(n)$.

We invoke $\textsc{ConstructTheGraph}$, Step 1, Step 2 and $\textsc{ConstructSuperstringByPath}$  sequentially. Therefore, the total complexity is the sum of complexities for these steps. So, the total complexity is $O^*(1.728^n)$.

Only Step 2 has an error probability. The most nested invocation of the Durr-Hoyer algorithm has an error probability $0.1$. Let us repeat it $2n$ times and choose the maximal value among all invocations. The algorithm has an error only if all invocations have an error. Therefore, the error probability is $0.1^{2n}=100^{-n}$.

Let us consider the middle Durr-Hoyer algorithm's invocation. The probability of success is the probability of correctness of maximum finding and the probability of input correctness, i.e., the correctness of all the nested Durr-Hoyer algorithm's invocations. It is
\[0.9\cdot  (1-100^{-n})^{\gamma}\mbox{, where }\gamma= \binom{n/2}{n/4}\]
\[\geq 0.8\mbox{, for enough big }n.\] 
So, the error probability is at most $0.2$.

 Let us repeat the middle Durr-Hoyer algorithm $2n$ times and choose the maximal value among all invocations. Similar to the previous analysis, the error probability is $0.2^{2n}=25^{-n}$.
 
 Therefore, the total success probability that is the final Durr-Hoyer algorithm's success probability is the following one.
 \[0.9\cdot  (1-25^{-n})^{\beta}\mbox{, where }\beta= \binom{n}{n/2}\]
 \[>0.8\mbox{, for enough big }n.\]
 Therefore, the total error probability is at most $0.2$.
\end{proof}

\section{Conclusion}\label{sec:conclusion}
We present a quantum algorithm for the SSP or SCS problem.  It works faster than existing classical algorithms. At the same time, there are faster classical algorithms in the case of restricted length of strings \cite{gkm2013,gkm2014}. It is interesting to explore quantum algorithms for such a restriction. Can quantum algorithms be better than classical counterparts in this case?

Another open question is approximating algorithms for the problem. As we mentioned before, such algorithms are more useful in practice. So, it is interesting to investigate quantum algorithms that can be applied for practical cases.
%
%
%
 \bibliographystyle{splncs04}
 \bibliography{report}
\end{document}